\documentclass[10pt]{article}
\usepackage{geometry}

\usepackage{graphicx}        % standard LaTeX graphics tool
\usepackage{amsmath,amssymb,amsthm}
\usepackage{natbib}

\usepackage{subcaption}

\DeclareMathOperator{\sech}{sech}

\newtheorem{theorem}{Theorem}

\begin{document}

\makeatletter
	
\newcommand{\E}{\mathrm{e}\kern0.2pt} 
\newcommand{\D}{\mathrm{d}\kern0.2pt}
\newcommand{\ii}{\kern0.05em\mathrm{i}\kern0.05em}
	
\newcommand{\bareta}{\bar{\eta}} 
\newcommand{\hh}{h(s)}
	
\newcommand{\RR}{\mathbb{R}} 
\newcommand{\CC}{\mathbb{C}}
	
\def\bottomfraction{0.9}

\title{{\bf Solitary waves on rotational flows with an interior stagnation
point}}
	
\author{V. Kozlov$^1$, N. Kuznetsov$^2$, E. Lokharu$^1$}
	
\date{}
	
\maketitle
	
\vspace{-8mm}
	
\begin{center}
$^1${\it Department of Mathematics, Link\"oping University, S--581 83
Link\"oping, Sweden \\ $^2$Laboratory for Mathematical Modelling of Wave
Phenomena, \\ Institute for Problems in Mechanical Engineering, Russian Academy
of Sciences, \\ V.O., Bol'shoy pr. 61, St. Petersburg 199178, RF}
		
\vspace{2mm}
		
E-mail: vladimir.kozlov@liu.se; nikolay.g.kuznetsov@gmail.com;
evgeniy.lokharu@liu.se 
\end{center}
	
\begin{abstract}
The two-dimensional free-boundary problem describing steady gravity waves with
vorticity on water of finite depth is considered. Under the assumption that the
vorticity is a negative constant whose absolute value is sufficiently large, we
construct a solution with the following properties. The corresponding flow is
unidirectional at infinity and has a solitary wave of elevation as its upper
boundary; under this unidirectional flow, there is a bounded domain adjacent to
the bottom, which surrounds an interior stagnation point and is divided into two
subdomains with opposite directions of flow by a critical level curve connecting
two stagnation points on the bottom.
		
\vspace{2mm}
		
\end{abstract}

\section{Introduction}

\setcounter{equation}{0}

In the present paper, we consider the problem describing two-dimensional gravity
waves travelling on a flow of finite depth. For an ideal fluid of constant
density, say water, the effects of surface tension are neglected, whereas the
flow is assumed to be rotational with a constant vorticity; this, according to
observations, is the type of motion commonly occurring in nature (see, for
example, \cite{SwanCummingJames01,Thomas1981} and references cited therein).
Also, it is assumed that the reference frame is moving with the wave so that the
relative velocity field is stationary. Our aim is to consider a new class of
solitary waves each having a cat's-eye---a region of closed streamlines
surrounding a stagnation point. Previously, this kind of behaviour was known
only for periodic waves with vorticity.
	 
The mathematical theory of two-dimensional solitary waves on irrotational flows
goes back to the discovery of John Stcott Russell, who was
the first to observe in 1834 and subsequently to analyse a solitary wave of
elevation; see \cite{Russell1844}. The existence of the latter was justified mathematically by Boussinesq in 1877 and rediscovered by Korteweg and de Vries in 1895. The
existence of solitary waves in the framework of the full water wave problem is
far more complicated and the first proofs were obtained much later (by
\cite{Lavrentiev43} and  
\cite{FriedrichsHyers54}). Modern proofs by \cite{Beale77} and
\cite{Mielke88} use the Nash--Moser implicit function theorem and a
dynamical system approach respectively. All these papers deal only with waves of
small amplitude, whereas \cite{AmickToland81b} constructed
large-amplitude solitary waves using global bifurcation theory and then proved
the existence of a limiting wave of the extreme form, that is, having an angled
crest, see \cite{AmickFraenkelToland82}. All solitary waves considered in these
papers are of positive elevation, symmetric and monotone on each side of the
crest, see \cite{CraigStern88, mcleod_1984}. The corresponding flows, being
irrotatonal, have a simple structure of streamlines: they are unbounded curves
similar (diffeomorphic) to the free surface profile. For all
\textit{unidirectional} waves with vorticity (when the horizontal component $u$
of the relative velocity field has a constant sign everywhere in the fluid), the
latter property is also true. Essentially, this forbids the presence of
critical layers and stagnation points.

The first construction of unidirectional small-amplitude solitary waves
with vorticity was given by \cite{Ter-Krikorov1962}, whereas \cite{Benjamin62}
obtained an approximate form of the wave profile which is the same as 
in the irrotational case. However, the relationships between the wave amplitude, the length scale and the propagation velocities depend on the primary velocity distribution in a complicated way. Much later, Hur (2008) and Groves \& Wahl\'en
(2008) obtained new results on this topic. The latter authors also considered
solitary waves with vorticity in the presence of surface tension
\citet{GrovesWahlen07}. The method used in \cite{GrovesWahlen07,GrovesWahlen08},
known as \textit{spatial dynamics}, is essentially an infinite-dimensional
version of the centre-manifold reduction which is known as spatial dynamics
because is applied to a Hamiltonian system with the horizontal spatial
coordinate playing the role of time. The first use of this method in the
water-wave theory is due to \cite{Kirchgaessner88,
Kirchgaessner82} (see also \cite{Mielke1986,Mielke88,Mielke1991}),
whereas an application of spatial dynamics to three-dimensional waves is given
in \cite{Groves2018} (see also references cited therein). So far, use of spatial
dynamics was restricted exclusively to small-amplitude waves. Recently, 
\cite{Wheeler13} examined waves of large amplitude, but, like in the
irrotational case, all solitary-wave solutions have the same structure of
streamlines, that is, are symmetric and of positive elevation; see \cite{Hur08,
Wheeler15b, KozKuzLok_2015, KozKuzLok_2017}. Thus, looking for a more
complicated geometry of solitary waves, it is natural to consider flows with
stagnation points and critical levels within the fluid domain. As in
\cite{Wahlen09}, by a {\it critical level} we mean a curve, where the horizontal
component of velocity vanishes.

The simplest case of flows with critical levels is that of constant negative
vorticity, and this has several advantages. Flows with constant vorticity are
easier tractable mathematically and they are of substantial practical importance
being pertinent to a wide range of hydrodynamic phenomena (see
\cite{Constantin2016}, p.~196). A new feature of laminar flows with constant
vorticity (compared with irrotational ones) is that there are flows with
critical levels. Moreover, it was shown by \cite{Wahlen09} that small
perturbations of these parallel flows are periodic waves with arrays of
cat's-eye vortices---regions of closed streamlines surrounding stagnation
points. An extension to periodic waves of large amplitude with critical layers
is given in \cite{Constantin2016} (it includes overhanging waves). Despite the
fact that waves considered in \cite{Wahlen09} and \cite{Constantin2016} have
critical layers, the geometry of free surface profiles is still simple:
symmetric about every crest and trough and monotone in between (just like that
of the classical Stokes waves). Examples of more complicated wave profiles are
known (see \cite{EhrnstromEscherWahlen11, Aasen2017, Kozlov2017, Kozlov2019}),
but for their construction the vorticity distribution must be at least linear.
It should be emphasised that all theoretical studies of waves with critical
layers were restricted so far to the periodic setting.

In the 1980s and 1990s, much attention was devoted to numerical computation of
various solitary waves on flows with constant vorticity; see the papers
\cite{Vanden-Broeck1994, VandenBroeck94New} and references
cited therein. In particular, an interesting family of solitary-wave profiles
was obtained in the second of these papers; it approaches a singular one with
trapped circular bulb at the crest which happens as the gravity acceleration
tends to zero. However, no attempt was made to find bottom or interior
stagnation points.

    \begin{figure}[t!]
        \centering
        \includegraphics[scale=0.6]{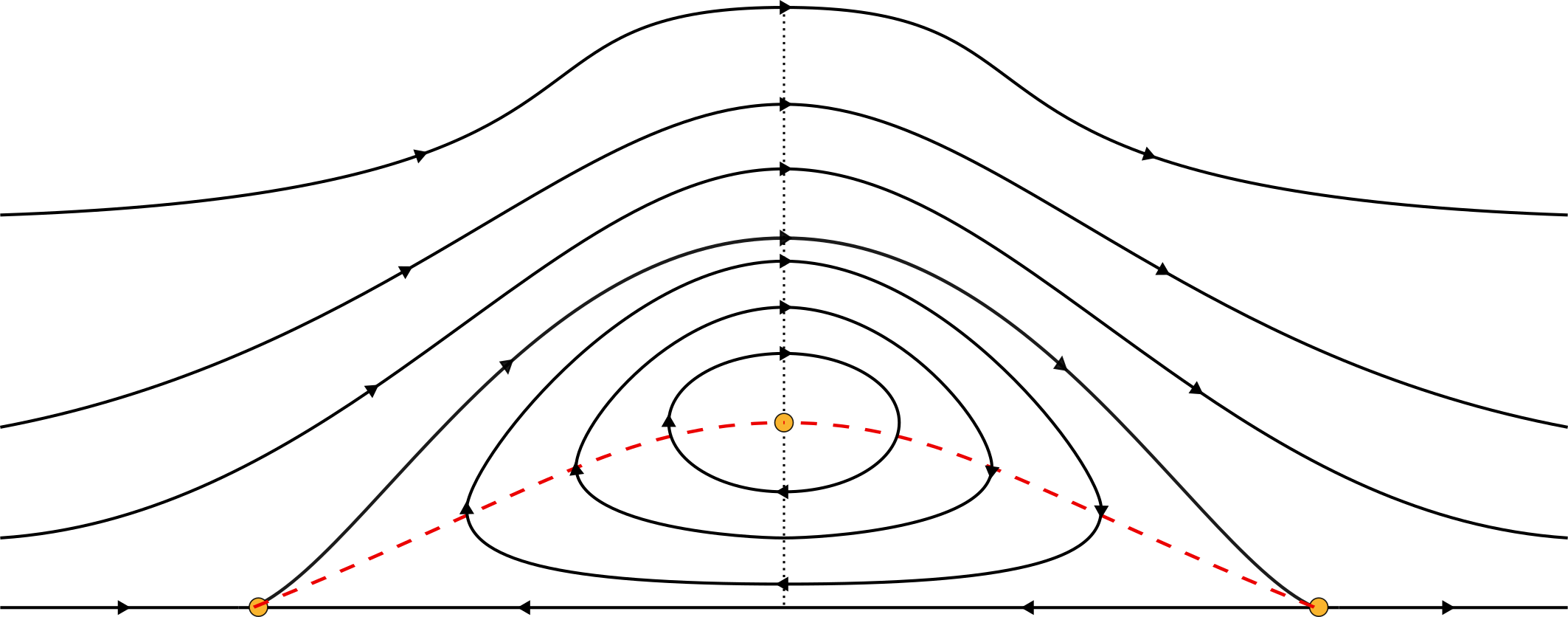}
        \caption{A sketch of the streamline pattern. Solid dots denote stagnation
       points and the dashed line shows the critical level, where the horizontal
       component of velocity vanishes. Another line connecting the bottom stagnation
       points is the critical streamline; the direction of streaming is indicated by 
       arrows.}
        \label{fig:solstreamlines}
    \end{figure}

In the present paper, a new family of solitary waves is constructed for large
negative values of the constant vorticity. All these waves have a remarkable
property: the corresponding flow is unidirectional at both infinities, but there
is a cat's-eye vortex centred below the wave crest; see Figure
\ref{fig:solstreamlines}. (The term was coined by Kelvin in his considerations
of a shear flow having this pattern of streamlines; see \cite{MajdaBertozzi02},
pp.~53--54.) The vortex is bottom-adjacent and separated from the unidirectional
flow above it by a {\it critical streamline} connecting two stagnation points on
the bottom. Every solitary wave under consideration is obtained as a long-wave
limit of Stokes wave-trains; in this aspect, our result is similar to that of
 \cite{AmickToland81a}, who dealt with the irrotational case.
However, Stokes waves has a cat's-eye vortex centred below each crest in our
case. For a sketch of the corresponding streamline pattern see the top Figure~4
below, whereas examples computed numerically are presented in
\cite{Ribeiro2017}, pp.~803--804. As the wavelength goes to infinity, these
vortices does not shrink, which is different from the case of small-amplitude
waves, whose sketch is plotted by \cite{Wahlen09} in his Figure~1. A
sketch of streamlines corresponding to our solitary wave is plotted in
Figure~\ref{fig:solstreamlines}, where the solid dots denote stagnation points.
Moreover, the dashed line shows the critical level along which the horizontal
component of velocity vanishes and the critical streamline located above the
critical level also connects the two bottom stagnation points; the direction of
streaming is indicated by arrows.
	
It should be emphasised that small-amplitude solitary waves constructed in this
paper cannot be captured by applying spatial dynamics directly because in a
certain sense the problem turns into a singular one for large values of the
vorticity. Thus, an appropriate scaling and a careful analysis are required
prior spatial dynamics can be used.
	
The plan of the paper is as follows. Statement of the problem and formulation of
main result are given in Subsection 1.1. Then, in Section 2, the problem is
scaled and reformulated in a suitable way. After that, in Section 3, it is
reduced to a finite-dimensional Hamiltonian system and Theorem 3 provides the
existence of solitary waves. Then the main Theorem 1 is proved in Section 4.

\subsection{Statement of the problem and formulation of the main result}
	
Let an open channel of uniform rectangular cross-section be bounded from below
by a horizontal rigid bottom and let water occupying the channel be bounded from
above by a free surface not touching the bottom. The surface tension is
neglected on the free surface, whereas the pressure is assumed to be cons tant
there. In appropriate Cartesian coordinates $(X,Y)$, the bottom coincides with
the $X$-axis and gravity acts in the negative $Y$-direction. We choose the frame
of reference so that the velocity field is time-independent as well as the
unknown free-surface profile. The latter is assumed to be the graph of $Y = \eta
(X)$, $X\in \RR$, where $\eta$ is a positive function. The water motion is
supposed to be two-dimensional and rotational; combining this and the
incompressibility of water, we seek the velocity field in the form $(\psi_Y,
-\psi_X)$, in which case $\psi (X,Y)$ is referred to as the stream function.

It is convenient to use the non-dimensional variables proposed by  \cite{KeadyNorbury78}. Namely, lengths and velocities are scaled to
$(Q^2/g)^{1/3}$ and $(Q g)^{1/3}$ respectively, where $Q$ is the mass flux and
$g$ is the acceleration due to gravity respectively. Thus, $Q$ and $g$ are
equal to unity in this variables.	Since the surface tension is neglected, the
pair $(\psi , \eta)$ must satisfy the following free-boundary problem:\\[-6mm]
\begin{subequations}\label{eqn:streamsystem}
	\begin{alignat}{2}
	 \psi_{XX} + \psi_{YY} - b & = 0 &\qquad& \text{for } 0 < Y < \eta (X), \label{lapp}
     \\ \psi (X,Y) & = 0 &\quad& \text{on} \ Y=0, \label{bcp} \\ 
	\psi (X,Y) & = 1 &\quad& \text{on} \ Y=\eta, \label{kcp} \\ 
	|\nabla \psi (X,Y)|^2 + 2 Y & = R &\quad& \text{on} \ Y=\eta (X). \label{bep}
	\end{alignat}
\end{subequations}
Here $R$ is a constant considered as problem's parameter; it is referred to as
the total head or the Bernoulli constant (see, for example,
\cite{KeadyNorbury78}). This statement (with a general vorticity distribution)
has long been known and its derivation from the governing equations and the
assumptions about the boundary behaviour of water particles can be found in
\cite{ConstantinStrauss04}.
	
A solution of problem \eqref{lapp}--\eqref{bep} defines a solitary wave provided
the following relations hold 
\begin{equation}
\eta(X) \to h \ \ \text{and} \ \ |\psi_X (X,Y)| \to 0 \ \ \text{as} \ \ X \to
\pm \infty . \label{sw}
\end{equation}
Here $h$ is a constant, which coincides with the depth of a certian laminar flow
at infinity. A sketch of the profile, that is typical for a solitary wave, is
shown in Figure 2. It should be noted that the flow at infinity is not uniform
as it is in the irrotational case.

\begin{figure}[!t]
	\centering
	\includegraphics[scale=0.6]{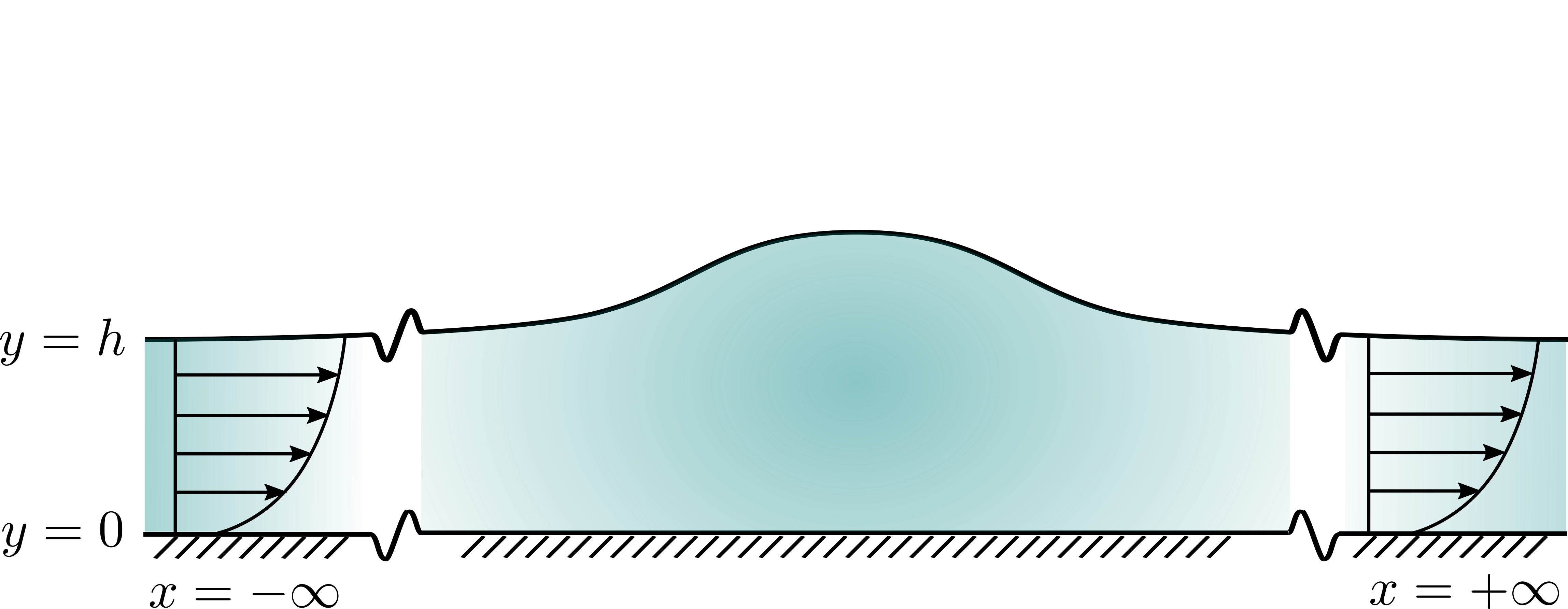}
	\caption{A sketch of the solitary wave profile on a unidirectional flow.}
	\label{fig:solitary}
\end{figure}	
 
Now we are in a position to formulate our main result about
the existence of solitary waves of elevation.
	
\begin{theorem} {\it For every sufficiently large $b > 0$ problem 
\eqref{lapp}--\eqref{bep}, \eqref{sw} has a solution $(\psi, \eta)$ with the
following properties:
	
{\rm (i)} $\eta(X) > h$ for all $X \in \RR$, that is, $\eta$ describes a
solitary wave of elevation;
	
{\rm (ii)} there are two stagnation points on the bottom and two streamlines
within the fluid domain connect these points (see Figure 1); the critical level
corresponds to the lower streamline, whereas the upper one of these streamlines
is critical;
	
{\rm (iii)} all streamlines below the critical one are closed and surround an
interior stagnation point on the vertical line through the crest;

{\rm (iv)} all streamlines above the critical one are diffeomorphic to the free
surface profile. }
\end{theorem}

An equivalent formulation of this assertion and its proof are given in the next
two sections. Our approach is based on a carefully chosen scaling of the
original problem. Then we apply the spatial dynamics method to the scaled
problem in the same way as in \cite{Kozlov2019}. This allows us to reduce the
problem to a finite-dimensional Hamiltonian system; it has one degree of freedom
and admits a homoclinic orbit describing a solitary wave of elevation in the
original coordinates. The orbit goes around an equilibrium point representing a
shear flow of constant depth with a counter-current and this guarantees the
presence of a stagnation point and a critical streamline as is illustrated in
Figure 1 above.

\section{Reformulation of the problem}

To avoid difficulties arising from the fact that $b$ is large, it is convenient
to scale variables as follows:
\[
\bar{x} = \sqrt{b} X, \ \ \bar{y} = \sqrt{b} Y, \ \ \bareta (\bar{x}) = \sqrt{b}
\, \eta(X), \ \ \bar{\psi}(\bar{x},\bar{y}) = \psi(X,Y) .
\]
This transforms \eqref{lapp}--\eqref{bep} into
\begin{subequations}\label{eqn:scaled}
	\begin{alignat}{2}
	\bar{\psi}_{\bar{x}\bar{x}} + \bar{\psi}_{\bar{y}\bar{y}} - 1 & = 0 &\qquad& 
    \text{for } 0 < \bar{y} < \bareta(\bar{x}), \label{eqn:scaled:1} \\  
	\bar{\psi} (\bar{x},\bar{y}) & = 0 &\quad& \text{on } \bar{y}=0, \label{eqn:scaled:2}
    \\ \bar{\psi} (\bar{x},\bar{y}) & = 1 &\quad& \text{on } \bar{y}=\bareta,
	\label{eqn:scaled:3} \\ |\nabla \bar{\psi}|^2 + 2 \gamma \bar{y} & = \bar{R} 
    &\quad& \text{on }\bar{y}=\bareta , \label{eqn:scaled:4}
	\end{alignat}
\end{subequations}
where $\gamma = b^{-3/2}$ and $\bar{R} = R b^{-1}$. This problem describes
two-dimensional waves with vorticity (the latter is equal to one) and weak
gravity because $\gamma$ is a small parameter provided $b$ is large.

Let us consider the stream solution $\bar{\psi} = u(\bar{y};s)$ and $\bareta =
h(s)$ such that $u'(0)=s$. From \eqref{eqn:scaled:1}--\eqref{eqn:scaled:3} one
obtains the unique pair
\begin{equation}\label{eqn:laminar}
u(\bar{y};s) = \tfrac 12 \bar{y}^2 + s\bar{y}, \ \ h(s) = -s + \sqrt{2+s^2} \,,
\end{equation}
and \eqref{eqn:scaled:4} yields the corresponding Bernoulli constant:
\[
\bar{R}(s) = 2 \gamma h(s) + [u' (h(s);s)]^2 .
\]
If $s<0$, then the laminar flow defined by \eqref{eqn:laminar} has a
counter-current, whereas the corresponding flow is unidirectional when $s > 0$.
In what follows we assume that $s<0$.

\subsection{Flattening transformation}

Changing the coordinates $(\bar{x}, \bar{y})$ to
\[ 
  (x, y) = \left( \bar{x},  \frac{\bar{y}}{\bareta(\bar{x})}h(s) \right) ,
\]
we map the water domain onto the strip $\RR \times (0, h(s))$. Let
\[
\hat \Phi (x, y) = \bar{\psi} \left( \bar{x} , \bar{y} \right)
\]
be new unknown function, for which problem \eqref{eqn:scaled} with $\bar{R} =
\bar{R}(s)$ takes the form:
\begin{subequations}\label{eqn:sysflat}
	\begin{alignat}{2}
	\Big[ \hat{\Phi}_x - \frac{y \bareta_x} {\bar \eta} \hat{\Phi}_y \Big]_x
	 - \frac{y \bareta_x} {\bareta} 	\Big[ \hat{\Phi}_x - \frac{y \bareta_x} 
{\bar \eta} \hat{\Phi}_y \Big]_y  + \frac{h^2(s)}{\bar{ \eta}^2}\hat{\Phi}_{yy} - 1
 & = 0 &\qquad& \text{for } 0 < y < h(s), \label{eqn:sysflat:1} \\  
	\hat{\Phi} (x,0) & = 0 &\quad& \text{for }x\in\RR, \label{eqn:sysflat:2} \\ 
	\hat{\Phi} (x,h(s)) & = 1 &\quad& \text{for }x\in\RR,
	\label{eqn:sysflat:3} \\ 
	\hat{\Phi}_y^2 (x, h(s))-\frac{ \bareta^2 (\bar{R}(s) - 2 \gamma \bareta)} 
{h(s)^2(1 + \bareta_x^2)} & = 0 &\quad& \text{for }x\in\RR. \label{eqn:sysflat:4}
	\end{alignat}
\end{subequations}
Note that $\hat{\Phi} = u(y;s)$ and $\bareta = h(s)$ is a solution of this
system. Let us write \eqref{eqn:sysflat} as a first-order system, for which
purpose it is convenient to introduce the variable $\hat{\Psi}$ conjugate to
$\hat \Phi$ (cf. \cite{Kozlov2013}):
\begin{equation}\label{eqn:defhatpsi}
\hat{\Psi} (x, y) = \frac{\bareta}{h(s)} \Big[ \hat{\Phi}_x - \frac{y \bareta_x}
{\bareta} \hat{\Phi}_y \Big] .
\end{equation}
This allows us to write \eqref{eqn:sysflat} as follows:
\begin{subequations}\label{eqn:sysvect}
	\begin{alignat}{2}
& \hat{\Phi}_x  = \frac{h(s)}{\bareta} \hat{\Psi} + \frac{y}{\bareta} \bareta_x
 \hat{\Phi}_y &\quad&\text{for } (x, y) \in \RR \times (0, h(s)), \label{eqn:sysvect:1}
 \\
& \hat{\Psi}_x  =  \frac {\bareta_x}{\bareta} (y \hat{\Psi})_y -
\frac{h(s)}{\bareta} \hat{\Phi}_{yy} + \frac{\bareta}{h(s)}  &\quad&\text{for }
(x, y) \in \RR \times (0, h(s)), \label{eqn:sysvect:2}\\ & \hat{\Phi}(x,0)
= \hat{\Psi}(x,0) = 0 &\quad&\text{for } x\in \RR, \label{eqn:sysvect:3} \\ &
\hat{\Phi}(x,y)  =1 &\quad&\text{on } y=h(s),\label{eqn:sysvect:4}\\ &
\hat{\Phi}_y^2+\hat{\Psi}^2  =
\frac{\bareta^2}{h^2(s)}(\bar{R}(s)-2\gamma\bareta)&\quad&\text{on }
y=h(s).\label{eqn:sysvect:5}
	\end{alignat}
\end{subequations}
Furthermore, we have that
\begin{equation}\label{eqn:etaprime}
\bareta_x(x) = - \frac{\hat{\Psi}(x,h(s))}{\hat{\Phi}_y(x,h(s))}
\end{equation}

Relations \eqref{eqn:sysvect} can be considered as an infinite-dimensional
dynamical system for $\hat{\Phi}$ and $\hat{\Psi}$ only. Indeed, $\bareta$ and
$\bareta_x$ can be eliminated with the help of \eqref{eqn:sysvect:5} and
\eqref{eqn:etaprime}, which will be formalised in the next section.

\subsection{Linearization around a laminar flow}

Let us linearize relations \eqref{eqn:sysvect} around the stream solution
$\hat{\Phi} = u(y;s), \hat{\Psi} = 0, \bareta=h(s)$, for which purpose we
introduce
\[
\Phi = \hat{\Phi}-u- \frac{y u_y}{h(s)} \zeta  , \ \ \Psi = \hat{\Psi}, \ \
\zeta = \bareta - h(s).
\]	
Then we obtain from \eqref{eqn:sysvect}:
\begin{subequations}\label{eqn:finalsys}
	\begin{alignat}{2}
	& \Phi_x  = \Psi + N_1 &\quad&\text{for } (x, y) \in \RR \times (0, h(s)),
    \label{eqn:finalsys:1} \\
	& \Psi_x  =  - \Phi_{yy} + N_2  &\quad&\text{for } (x, y) \in \RR \times (0, h(s)),
    \label{eqn:finalsys:2}\\
	& \Phi(x,0)  =\Psi(x,0) = 0 &\quad&\text{for } x\in \RR, \label{eqn:finalsys:3} \\
	& \Phi_y - \kappa \Phi  = N_3 &\quad&\text{on } y=h(s). \label{eqn:finalsys:4}
	\end{alignat}
\end{subequations}
Here
\[
\kappa = \kappa(s,\gamma) = \frac{\gamma+k}{k^2} \ \mbox{and} \ k = k(s) =
h(s)+s = \sqrt{2+s^2} > 0 \, ,
\]
whereas the nonlinear operators in \eqref{eqn:finalsys:1},
\eqref{eqn:finalsys:2} and \eqref{eqn:finalsys:3} have the form
\[
\begin{split}
N_1 & = \frac{-h(s) \Psi \zeta + y \zeta_x(y\zeta + h(s)
\Phi_y)}{h(s)(h(s)+\zeta)},\\ N_2 & = \frac{\zeta^2 + h(s)\zeta_x(y\Psi)_y +
h(s)\zeta \Phi_{yy}}{h(s)(h(s)+\zeta)},\\ N_3 & = \frac{-\hh^2\Psi^2 +
(\zeta+\Phi_y)(-\hh\zeta(\hh-2k) +2\zeta^2 k-\hh^2\Phi_y)}{2(\hh+\zeta)^2 k} ,
\end{split}
\]
respectively. Moreover, we find that
\begin{equation}\label{zetazetaprime}
\zeta(x) = - \frac{\Phi(x,h(s))}{k}, \ \ \zeta_x(x) = - \frac{\hh \Psi(x,h(s))}{\hh 
\zeta + \hh k + \zeta k + \hh \Phi_y(x,h(s))}.
\end{equation}
Substituting these expressions into formulae for $N_1$, $N_2$ and $N_3$, we see
that \eqref{eqn:finalsys:1} and \eqref{eqn:finalsys:2} form an
infinite-dimensional \textit{reversible} dynamical system on the manifold
defined by \eqref{eqn:finalsys:3} and \eqref{eqn:finalsys:4}. The presence of
the nonlinear boundary condition \eqref{eqn:finalsys:4} is inessential because
it is reducible to a homogeneous one after a proper change of variables; see
\cite{GrovesWahlen08} and \cite{Kozlov2019} for details.

Let us assume the following regularity
\[
\Psi \in C(\RR; X_1), \ \ \Phi \in C(\RR;X_2) ,
\] 
where 
\[
X_m = \{ f \in H^m(0,1): \  f(0) = 0 \}, \ \ \ m=0,1,2,
\] 
and $H^m(0,1)$ denotes the Sobolev space. Furthermore, it is clear that $k$ and
$\kappa$ depend analytically on $s$ and $\gamma$, provided their values are
small, and so the same is true for the operators $N_1$, $N_2$ and $N_3$. More
precisely, let
\[ \Lambda_\epsilon = \{ \lambda = (s, \gamma) \in \RR^2: |s|^2 + |\gamma|^2 < 
\epsilon^2 \} 
\]
be a small neighbourhood of the origin in the parameter space, then 
\[ N_1 \in C^{\infty} (X_1
\times X_2 \times \Lambda_\epsilon; H^1(0,1)) , \quad N_2 \in C^{\infty} (X_1
\times X_2 \times \Lambda_\epsilon; L^2(0,1)) ,
\]
whereas $N_3 \in C^{\infty} (X_1 \times X_2 \times \Lambda_\epsilon; \RR)$.
Moreover, all derivatives of these operators are bounded and uniformly
continuous in $\Lambda_\epsilon$.

\subsection{A linear eigenvalue problem}

The centre subspace of \eqref{eqn:finalsys} is determined by the imaginary
spectrum of the linear operator $L (\Psi, \Phi)=(- \Phi_{YY}, \Psi)$ defined on
a subspace of $X_1 \times X_2$ and subject to the homogeneous condition
\[
\Phi_Y (h(s)) = \kappa \Phi (h(s)).
\]
It is straightforward to find that the spectrum is discrete and consists of all
$\hat{\tau} \in {\mathbb C}$ such that $\mu = \hat{\tau}^2$ is an eigenvalue of the
following Sturm--Liouville problem:
\begin{equation}
- \varphi_{YY} = \mu \varphi \ \mbox{on} \ (0, 1) ; \quad \varphi (0) = 0 \
\mbox{and} \ [ \varphi_Y - \kappa \varphi ]_{Y=1} = 0 \, . \label{SL}
\end{equation}
(Basic facts about Sturm--Liouville problems can be found in \cite{Teschl12}.)
Thus, the imaginary part of the spectrum of $L$ corresponds to the negative
eigenvalues of \eqref{SL}.

The spectrum of \eqref{SL} is discrete and consists of real simple eigenvalues
say 
\[
\mu_1 < \mu_2 <  \cdots < \mu_j < \cdots
\]
accumulating at infinity, and the corresponding eigenfunctions $\varphi_j$ can
be rescaled to form an orthonormal basis for $L^2(0,h(s))$.

\subsection{On the existence of a negative eigenvalue}
	
Let us investigate the spectral problem \eqref{SL} for small negative $s$ and
small positive $\gamma$. Solving \eqref{SL} explicitly, we find that the unique
negative eigenvalue $\mu_1 = - \tau^2$ satisfies the dispersion equation:
\[
\tau \hh \coth(\tau \hh) = \kappa \hh .
\]
Using the definition of $\kappa$, we obtain that
\[
\kappa \hh = 1 + \frac{\gamma-s}{\sqrt{2}} + O (\epsilon^2) \ \ \mbox{as} \
\epsilon \to 0 .
\]
Thus, $\kappa \hh > 1$ for all sufficiently small negative $s$ and small $\gamma
> 0$, and so the dispersion equation has a unique root $\tau> 0$ such that
\begin{equation}\label{tau}
[\tau \hh]^2 = 3\frac{\gamma-s}{\sqrt{2}} + O(\epsilon^2) \ \ \mbox{as} \
\epsilon \to 0 .
\end{equation}
Since $h(s)=\sqrt{2}+O(s)$, we have the following approximate expression for the
normalized eigenfunction $\phi_1$:
\begin{equation}\label{c0}
\phi_1(y) = c_0 \, y (1 + O(\epsilon)) \ \ \mbox{as} \ \epsilon \to 0 , \
\mbox{where} \ c_0 = \sqrt{\frac{3}{2^{3/2}}} \, .
\end{equation}
It is uniform with respect to $ y \in [0,\hh]$ and $\lambda = (s,\gamma) \in
\Lambda_\epsilon$.
	
Finally, it is worth mentioning that the positive part of the spectrum of $L$ is
separated from zero, because $\mu_2 > \pi^2/h^2(s)> \pi^2/2$.

\section{Reduction to a finite-dimensional system} 
	
Let us reduce \eqref{eqn:finalsys} to a finite-dimensional Hamiltonian system by
using the centre-manifold technique developed by \cite{Mielke88}, who
considered quasilinear elliptic problems in cylinders. For this purpose we apply
a result obtained in \cite{Kozlov2019} (namely, Theorem 3.1), which is a
convenient way to obtain a reduced problem. Prior to that we use the so-called
spectral splitting to decompose the system.

\subsection{Spectral decomposition and reduction}
	
Following the method proposed in \cite{Kozlov2019}, we seek $(\Phi, \Psi)$ in
the form
\[ \Phi (x, y) = \alpha (x) \, \varphi_1 (y) + \widetilde \Phi (x, y) , \quad
\Psi (x, y) = \beta (x) \, \varphi_1 (y) + \widetilde \Psi (x, y) \, ,
\label{Phi}
\]
where $\widetilde{\Phi}$ and $\widetilde{\Psi}$ are orthogonal to $\varphi_1$
in $L^2(0,\hh)$; that is,
\[ \alpha (x) = \int_0^{\hh} \Phi (x, y) \, \varphi_1 (y) \, \D y , \quad \beta (x) =
\int_0^{\hh} \Psi (x, y) \, \varphi_1 (y) \, \D y \, .
\]
For $\lambda \in \Lambda_\epsilon$ we define projectors ${\mathcal P}_{\lambda}
\phi = \alpha \varphi_1$ and $\widetilde{\mathcal P}_{\lambda} = {\rm id} -
{\mathcal P}$, which are well defined on $H^1 (0,\hh)$ and orthogonal in $L^2
(0,\hh)$. Multiplying \eqref{eqn:finalsys:1} and \eqref{eqn:finalsys:2} by
$\varphi_1$ and integrating over $(0,\hh)$, we obtain
\begin{align}
& \alpha_X = \beta + F_{1}(\Psi,\Phi; \lambda) , \label{red1} \\ & \beta_X =
-\tau^2 \alpha + F_{2}(\Psi,\Phi; \lambda) , \label{red2}
\end{align}
where
\[ 
\begin{split}
& F_{1}(\Psi,\Phi; \lambda) = \int_0^{\hh} N_1(\Psi,\Phi; \lambda) \varphi_1 \, \D y , \\ 
& F_{2}(\Psi,\Phi; \lambda) = \int_0^{\hh} N_2(\Psi,\Phi; \lambda) \varphi_1 \, \D y - 
N_3(\Psi,\Phi; \lambda) \, \varphi_1 (\hh) \, .
\end{split} these
\]
The system for $\widetilde{\Phi}$ and $\widetilde{\Psi}$ is as follows:
\begin{eqnarray}
&& \widetilde{\Phi}_x = \widetilde{\Psi} + \widetilde{\mathcal P}_{\lambda} (N_1),
\label{red3} \\ && \widetilde{\Psi}_x = - \widetilde{\Phi}_{yy} + \widetilde{\mathcal
P}_{\lambda} (N_2) + \varphi_1 (\hh) \, \varphi_1 N_3 , \label{red4}
\end{eqnarray} 
and these functions satisfy the following boundary conditions:
\begin{equation}
\widetilde{\Phi} (x,0) = \widetilde{\Psi}(x,0) = 0 , \quad \widetilde{\Phi}_Y(x,\hh) -
\kappa \, \widetilde{\Phi} (x,\hh) = N_3 \, . \label{red5}
\end{equation}
Moreover, let $\widetilde{X}_{j}^{(\lambda)}$ denote $\widetilde{\mathcal
P}_{\lambda}(X_j)$, $j=0,1,2$, where $\lambda\in \Lambda_\epsilon$, then for all
$x \in \RR$ and $\lambda\in \Lambda_\epsilon$ we have that $\widetilde{\Psi}
\in \widetilde{X}_{1}^{(\lambda)}$ and $\widetilde{\Phi} \in
\widetilde{X}_{2}^{(\lambda)}$.

Applying Theorem 3.1 proved in \cite{Kozlov2019} to the decomposed system
\eqref{red1}--\eqref{red5}, we arrive at the following assertion.
	
\begin{theorem}\label{Tred} For any $m \geq 2$ there exist $\epsilon > 0$,
neighbourhoods $W \subset \RR^2,$ $W_1 \subset X_1, W_2 \subset X_2$ and the
vector-functions $r_j : W \times \Lambda_\epsilon \to W_j,$ $j=1,2,$ of the
class $C^m (W \times \Lambda_\epsilon)$ with the following properties.
\begin{itemize}
	\item[(I)] The derivatives of $r_1$ and $r_2$ are bounded and uniformly continuous,
    and the estimate
	\[
	\|r_1;H^1\|+\|r_2;H^2\|=O(|\alpha|^2+|\beta|^2), \ \ where \ (\alpha,\beta) \in W,
	\]
	holds uniformly with respect to $\lambda \in \Lambda_\epsilon$.
	\item[(II)] $r_j(\alpha,\beta,\lambda) \in \widetilde{X}_{j}^{(\lambda)},$ $j=1,2,$
    for all $\lambda \in \Lambda_\epsilon$ and all $(\alpha,\beta) \in W$.
	\item[(III)] The set
		\[
M^\lambda=\{ (\Psi^r[\alpha,\beta; \lambda],\Phi^r[\alpha,\beta;
\lambda])\,:\,(\alpha,\beta)\in W, \ \lambda \in \Lambda_\epsilon \} \subset X_1
\times X_2 ,
\]	
where $\Psi^r[\alpha,\beta; \lambda]=\beta\varphi_1+r_1(\alpha,\beta;\lambda)$
and $\Phi^r[\alpha,\beta; \lambda]=\alpha\varphi_1+r_2(\alpha,\beta;\lambda),$
is a locally invariant manifold for \eqref{eqn:finalsys}; that is, a unique
solution of \eqref{eqn:finalsys} goes through every point of $M^\lambda$ and
this solution belongs to $M^\lambda$ as long as $(r_1,r_2) \in W_1\times W_2$.
     \item[(IV)]
Every global solution $(\alpha,\beta) \in C(\RR;W)$ of the reduced system
\begin{equation} \label{eqn:reducedsystem}
\begin{split}
& \alpha_x = \beta + F_{1}(\Psi^r[\alpha,\beta; \lambda],\Phi^r[\alpha,\beta;
\lambda]; \lambda),   \\ & \beta_x = \mu_1 \alpha +  F_{2}(\Psi^r[\alpha,\beta;
\lambda],\Phi^r[\alpha,\beta; \lambda]; \lambda) ,
\end{split}
\end{equation}
where $\lambda \in \Lambda_\epsilon,$ generates the solution $(\Psi,\Phi)$ of
\eqref{eqn:finalsys} with
\[
\Psi(x,y) = \Psi^r[\alpha(x),\beta(x); \lambda](y), \ \ \Phi(x,y) =
\Phi^r[\alpha(x),\beta(x); \lambda](y).
\]
Moreover,  the reduced system \eqref{eqn:reducedsystem} is reversible.
\end{itemize}		
\end{theorem}
	
A direct calculation shows that the reduced system \eqref{eqn:reducedsystem} has
the following structure:
	\begin{equation} \label{redappr}
	\begin{split}
	&\alpha_x = \beta [ 1 + O(|\alpha|+|\beta|^2) ], \\
	&\beta_x = -\tau^2 \alpha + A \alpha^2 + O(|\alpha|^3+|\beta|^2),
	\end{split}
	\end{equation}
where $A = \tfrac 12 c_0^3+O(\epsilon)$	and $c_0$ is defined in \eqref{c0}. Now
we are in a position to formulate and prove the following.

\begin{theorem} \label{homsolthm} Problem \eqref{redappr} has a homoclinic
solution such that
\begin{equation} \label{homsoleq}
\alpha_{h}(x) = \tau^2\alpha^\star_+-\left(\tfrac{3}{2}\alpha_+^\star \tau^2 +
O(\tau^4) \right) \sech^2(\tau x/2) \ \ as \ \tau \to 0 ,
\end{equation}
whereas $\beta_h$ is defined implicitly by the first formula \eqref{redappr}.
Here $\alpha^\star_+ = 1/A + O(\tau^2)$ is a constant independent of $x$.
\end{theorem}

\begin{proof}	
It is known (see \cite{GrovesStylianou14},\cite{Kozlov2013}) that problem
\eqref{lapp}--\eqref{bep} has a Hamiltonian structure (even for arbitrary
vorticity) with the role of time played by the horizontal coordinate. The
Hamiltonian is the flow force invariant which has the following form in the
original coordinates $(X, Y)$:
	\[
	{\mathcal S} = \left[ \tfrac 12 R - b
	\right] \eta (X)  - \tfrac 12 \Big\{ \eta^2 (X) - \int_0^{\eta (X)}
	\tfrac 12 (\psi_Y^2 - \psi_X^2) + b \psi \, \D Y \Big\}.
	\]
Thus, the reduced system \eqref{redappr} has a constant of motion ${\mathcal
H}(\alpha,\beta)$, to obtain which one has to subject $b^{-1}{\mathcal S}$ to
all changes of variables described above. A direct calculation yields that
\[ \mathcal H (\alpha, \beta) = \frac{1}{2} (\beta^2 +\tau^2 \alpha^2) - \frac{A}{3}
\alpha^3 + O \left( \alpha^4 + \beta^2 \right) (= b^{-1}S) ,
\]
where $A$ is the coefficient in the second equation \eqref{redappr}. It should
be noted that $\mathcal H (\alpha, \beta)$ is an even function of $\beta$ which
follows from reversibility of this system. The form of $\mathcal H (\alpha,
\beta)$ suggests that variables must be scaled as follows:
\[
\alpha(x) = \tau^2 \alpha_1(x_1), \ \ \beta(x) = \tau^3 \beta_1(x_1), \ \
x=\tau^{-1} x_1 , \ \ {\mathcal H}(\alpha,\beta) = \tau^6{\mathcal
H}_1(\alpha_1,\beta_1),
\]
where
\[
{\mathcal H}_1(\alpha_1,\beta_1) = \frac{1}{2}(\alpha_1^2+\beta_1^2)-\frac{A}{3}
\alpha_1^3 + \tau^2O(|\alpha^3_1|+|\beta^4_1|) ,
\]
and so the scaled equations are
\begin{equation} \label{redappr1}
\begin{split}
&[\alpha_1]_{x_1} = \beta_1+ \tau^2\beta_1 O(|\alpha_1|+|\beta_1|^2), \\
&[\beta_1]_{x_1} = -\alpha_1 + A \alpha_1^2 +  \tau^2O(|\alpha_1|^3+|\beta_1|^2).
\end{split}
\end{equation}

\begin{figure}[t!]
	\centering
	\begin{subfigure}[t]{0.5\textwidth}
	\centering
\includegraphics[scale=0.8]{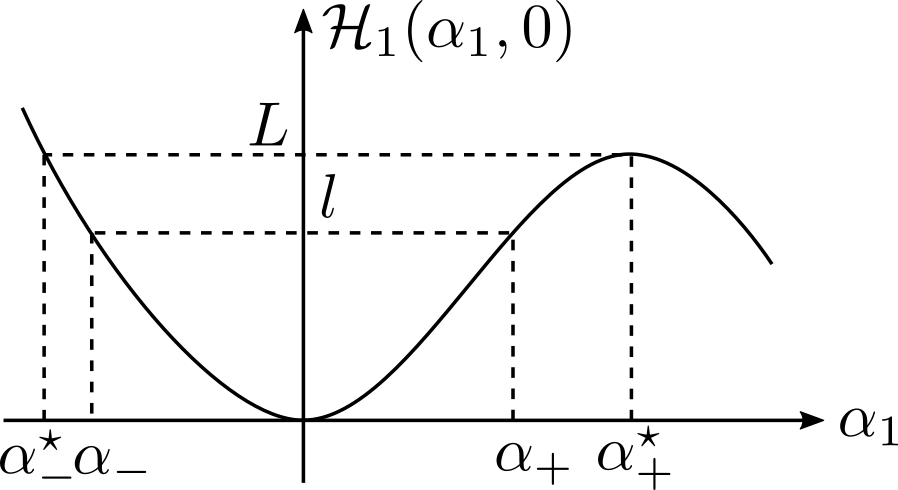}
\caption{Graph of the function ${\mathcal
H}_1(\alpha_1,0)$}
%\label{fig:H1}
	\end{subfigure}%
	~ 
	\begin{subfigure}[t]{0.5\textwidth}
	\centering
\includegraphics[scale=0.8]{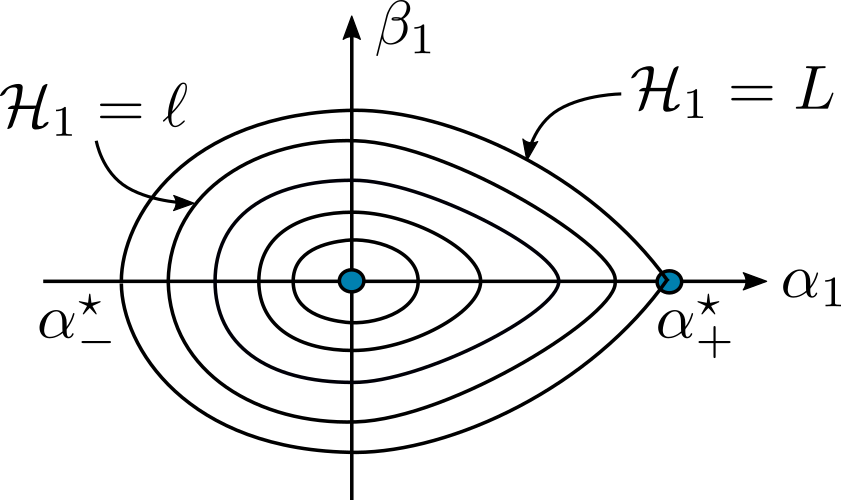} \caption{Phase portrait}
%\label{fig:phase}and
	\end{subfigure}
	\caption{The behaviour of the Hamiltonian near the origin (A) and its level curves (B)}
\end{figure}

The graph of ${\mathcal H}_1(\alpha,0)$ in a neighbourhood of the origin is
sketched in Figure 3A. It is clear that the local maximum of this function close
to the origin is attained at $\alpha=\alpha^\star_+$ and $L = {\mathcal
H}_1(\alpha^\star_+,0)$ is its value. Then the level line
\[
{\mathcal H}_1(\alpha_1,\beta_1) = \ell
\]
is a closed curve for every $\ell \in (0, L)$ and it corresponds to a periodic
solution. The contour
\[
{\mathcal H}_1(\alpha_1,\beta_1) = L
\]
defines the homoclinic orbit; see a sketch of level lines in Figure 3B. The
values $\alpha_-^\star$ and $\alpha_+^\star$ correspond to the ``crest'' and the
limiting depth respectively. An essential feature of the homoclinic orbit is
that $\alpha_1$ attains negative values on its left part; this implies that
there is a stagnation point as will be shown below.

Let us turn to proving \eqref{homsoleq}. By $\alpha_1^h$ we denote a homoclinic
solution to \eqref{redappr1}. It is easy to see that it is monotone on each side
of the crest; the latter corresponds to the value $\alpha_1^h(0) =
\alpha_-^\star$. Then we have
\[
X_1 = \int_0^{x_1} \textrm{d}x_1 = \int_{\alpha_-^\star}^{\alpha_1^h(x_1)}
\frac{\D \alpha_1}{[\alpha_1^h]_{x_1}}  ,
\]
where $[\alpha_1^h]_{x_1} = \beta_1^h [ 1+O(\tau^2) ]$; see the first equation
\eqref{redappr1}. On the other hand,
\[
L - {\mathcal H}_1(\alpha_1^h,0) = \frac{[\beta_1^h]^2}{2} \big\{ 1+\tau^2
O([\beta_1^h]^2) \big\} ,
\]
where the system's reversibility is used. Expressing $[\beta_1^h]^2$ from the
last formula and taking into account the fact it is positive on the interval of
the integration, we obtain
\begin{equation}\label{eqsol1}
x_1 = \int_{\alpha_-^\star}^{\alpha_1^h(x_1)} \frac{\D \alpha_1}
{[\alpha_1^h]_{x_1}} = [1+O(\tau^2)] \int_{\alpha_-^\star}^{\alpha_1^h(x_1)}
\frac{\D \alpha_1}{\sqrt{2 [L-{\mathcal H}_1(\alpha_1^h,0)]}} .
\end{equation}
Let us find an approximation of the integral using a third degree polynomial
for the expression under the square root; more precisely, let us show that
\begin{equation}\label{HamStruc}
\begin{split}
L-{\mathcal H}_1(\alpha_1,0)  = a(\tau) (\alpha_1-\alpha_-^\star)
(\alpha_1-\alpha_+^\star)^2 + \tau^2 O (|\alpha_1-\alpha_-^\star|
|\alpha_1-\alpha_+^\star|^2),
\end{split}
\end{equation}
where
\[
a(\tau) = - \frac{L}{\alpha_-^\star [\alpha_+^\star]^2} = \frac{A}{3}+O(\tau^2).
\]
It should be emphasised that \eqref{HamStruc} is used as a representation of
$L-{\mathcal H}_1(\alpha_1,0)$ only on the interval $[\alpha_-^\star,
\alpha_+^\star]$. To prove \eqref{HamStruc} we note that
\[
L-{\mathcal H}_1(\alpha_1,0) = L-\frac{\alpha_1^2}{2}+ A \frac{\alpha_1^3}{3} +
\tau^2 O(|\alpha_1|^3)=:Q_1(\alpha_1) + \tau^2 O(|\alpha_1|^3).
\]
A direct calculation yields that the estimate $O(\tau^2)$ holds for
\[
Q_1(\alpha_-^\star), \  Q_1(\alpha_+^\star), \ \partial_{\alpha_1}Q_1(\alpha_+^\star)
\]
Therefore, solving a linear system, one obtains that up to $O(\tau^2)$ the
coefficients of
\[
a(\tau)(\alpha_1-\alpha_-^\star)(\alpha_1-\alpha_+^\star)^2
\]
are the same as those of $Q_1$. This shows that the error in \eqref{HamStruc}
has the same estimate $O(\tau^2)$. It remains to use the fact that $L-{\mathcal
H}_1(\alpha_1,0)$ has a simple zero at $\alpha_-^\star$ and a double zero at
$\alpha_+^\star$ which proves \eqref{HamStruc}.

Now we have
\[
\begin{split}
\int_{\alpha_-^\star}^{\alpha_1^h(x_1)} \frac{\D \alpha_1}{\sqrt{2(L-{\mathcal
H}_1(\alpha_1,0))}} \! = & [1+O(\tau^2)] \int_{\alpha_-^\star}^{\alpha_1^h(x_1)}
\frac{\D \alpha_1}{\sqrt{2a(\alpha_1-
\alpha_-^\star)(\alpha_1-\alpha_+^\star)^2}} \\ = &
\frac{1+O(\tau^2)}{a_0\sqrt{2a}} \left[ \ln\left|
\frac{a_0+\sqrt{\alpha_1-\alpha_-^\star}}{a_0-\sqrt{\alpha_1-\alpha_-^\star}}
\right| \right]_{\alpha_1=\alpha_-^\star}^{\alpha_1=\alpha_1^h(x_1)} \\ = &
\frac{1+O(\tau^2)}{a_0\sqrt{2a}} \left( \ln\left[
\frac{\left(a_0+\sqrt{\alpha_1^h(x_1)-\alpha_-^\star}\right)^2}{\alpha_+^\star
-\alpha_1^h(x_1) }\right] \right),
\end{split}
\]
where $a_0 = \sqrt{\alpha_+^\star-\alpha_-^\star}$. Comparing this and
\eqref{eqsol1}, one obtains the following formula for the solitary wave solution:
\begin{equation} \label{solasymp}
\alpha_1^h(x_1) = \alpha_+^\star-\left[\frac{6}{A} +
O(\tau^2)\right]e^{-(1+O(\tau^2))x_1} ,
\end{equation}
where formulas $\alpha_-^\star = - \alpha_+^\star/2 + O(\tau^2)$ and
$\alpha_+^\star = 1/A + O(\tau^2)$ are taken into account. Furthermore, it is
straightforward to show that
\[
\|\alpha_1^h - \alpha_1^{h,\star}\|_{L^{\infty}(\RR)} = O(\tau^2),
\]
where 
\[
\alpha_1^{h,\star}(x_1) = \frac{1}{A} - \frac{3}{2A} \sech^2(x_1/2) .
\]
The latter is a homoclinic solution of \eqref{redappr1} with $\tau=0$, in which
case $\beta_1 = [\alpha_1^{h,\star}]_{x_1}$. Combining this and the asymptotic
formula \eqref{solasymp}, one arrives at \eqref{homsoleq} by rescaling variables
to the original ones.
\end{proof}

\begin{figure}[t!]
	\centering
	\includegraphics[scale=0.6]{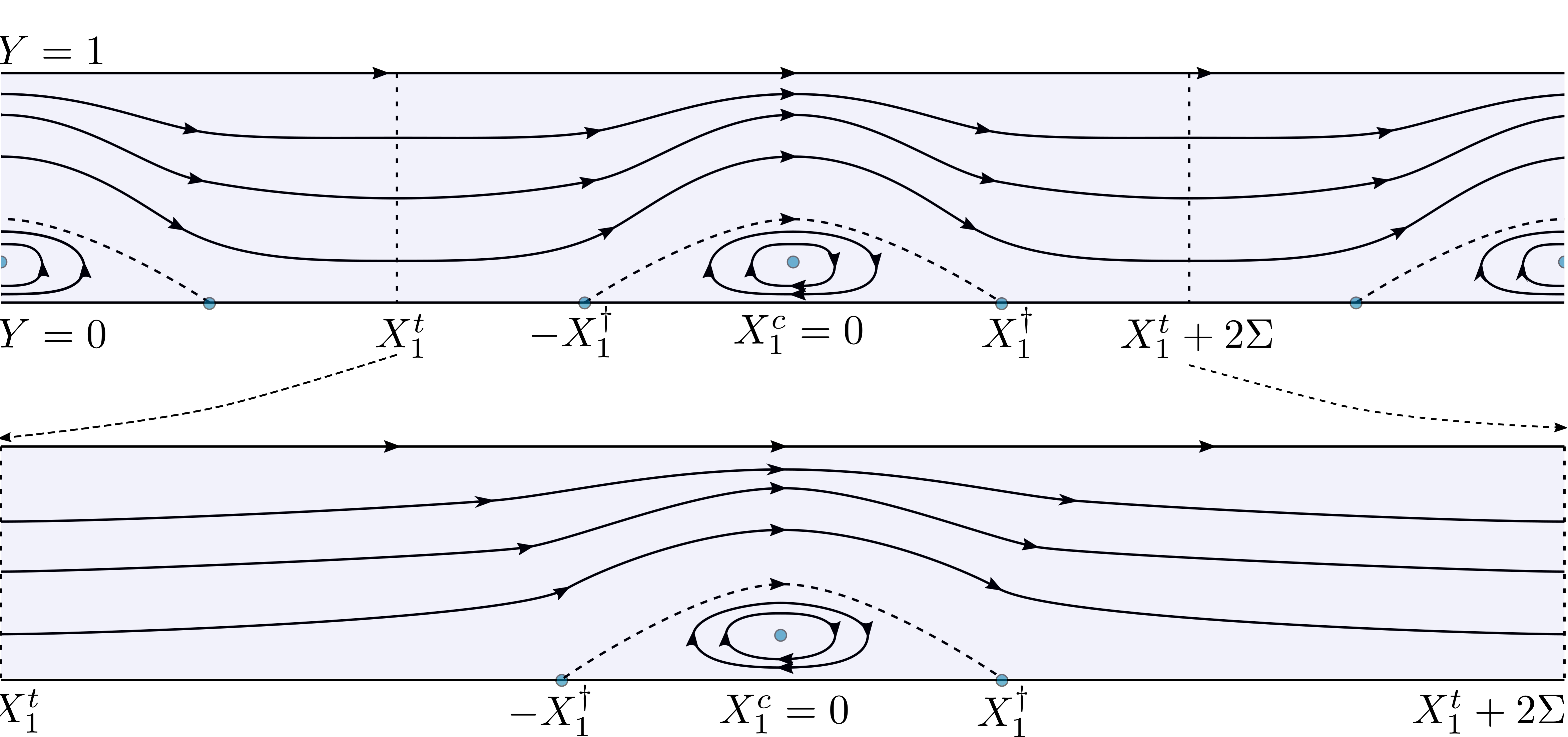}
	\caption{Two Stokes waves corresponding to different values of $\ell$ plotted
    in $(X_1,Y)$-variables.}
	\label{fig:stokes}
\end{figure}

\subsection{Periodic waves}

An approximation of solutions for periodic waves can be found in the same way as
for solitary waves. Indeed, let us consider the periodic solution
$\alpha_1^\ell$ corresponding to some energy ${\mathcal H}_1=\ell$. Let the
trough of the corresponding wave be located at $x_1^c = 0$ and let the nearest
crest to the left have the coordinate $x_1^t < x_1^c$ (see Figure
\ref{fig:stokes}). Then for every $x_1 \in (x_1^t x_1^c]$ we have
\begin{equation} \label{eq:periodic}
x_1 - x_1^t = \int_{\alpha_-}^{\alpha_1^\ell(x_1)} \!\! \frac{\D
\alpha_1}{[\alpha_1^\ell]_{x_1}} = [1+O(\tau^2)]
\int_{\alpha_-}^{\alpha_1^\ell(x_1)} \frac{\D \alpha_1}{\sqrt{2 [\ell-{\mathcal
H}_1(\alpha_1,0)]}} \, .
\end{equation}
In particular, the half-period of this solution is equal to
\begin{equation}
\Sigma_\ell = x_1^c-x_1^t = [1+O(\tau^2)] \int_{\alpha_-}^{\alpha_+} \!\!
\frac{\D \alpha_1}{\sqrt{2 [\ell - {\mathcal H}_1 (\alpha_1, 0) ]}} \, ,
\label{Lambda}
\end{equation}
and so $\Sigma_\ell \to +\infty$ as $\ell \to L$. Let us estimate the bottom
width of a cat's-eye vortex (in Figure \ref{fig:stokes}, it is bounded above by
the dashed streamline). On every interval symmetric about the crest and having
the length $2 \Sigma_\ell$, there are exactly two stagnation points on the
bottom that bound the bottom-attached vortex (these points are shown solid in
Figure \ref{fig:stokes}), and $\alpha_1^\ell(\pm x_1^\dagger) = 0$ at these
points nearest to the origin. From \eqref{eq:periodic} the approximate formula
follows:
\[
-x_1^\dagger-x_1^t = [1+O(\tau^2)] \int_{\alpha_-}^{0} \!\! \frac{\D
\alpha_1}{\sqrt{2 [\ell - {\mathcal H}_1 (\alpha_1, 0) ]}} \, ,
\]
which yields that $x_1^\dagger = O(1)$ as $\ell \to L$ in view that the
integral
\[
\int_{\alpha_-}^{0} \!\! \frac{\D \alpha_1}{\sqrt{2 [L - {\mathcal H}_1
(\alpha_1, 0) ]}}
\]
is finite. Indeed, the function $L - {\mathcal H}_1 (\alpha_1, 0)$ has only one
simple zero $\alpha_1=\alpha_-^\star$ on the interval of integration. Thus,
$x_1^\dagger$ remains bounded when the wavelength goes to infinity, and the same
is true for the domain occupied by cat's-eye vortex. Therefore, the structure of
streamlines is similar for the limiting solitary wave as shown in Figure
\ref{fig:solitary}.

\section{Proof of Theorem 1}

It is straightforward to recover the free surface profile which takes the form
\begin{equation}\label{solprofile}
\eta(X) = h_- + \sqrt{b} \left[ \frac{3}{c_0^2} + O(\tau^2) \right] \tau^2
\textrm{sech}^2(\sqrt{b} \tau X/2)
\end{equation}
in the original coordinates. Here $c_0^2 = 2 A k / \varphi_1 (h (s))$ and $h_-$
is the depth of the unidirectional laminar flow conjugate to $b^{-1/2}h(s)$.
Thus, $\eta(X)$ descibes a solitary wave of positive elevation.

To show that the flow, on which this wave propagates, has a bottom-attached
cat's-eye vortex, let us track back the changes of coordinates made above and
find that
\[
\hat{\Phi}_y = u_y + \Phi_y - (u_y + y) \Phi(x,h)/(k h) .
\]
Here $\Phi(x,y) = \alpha(x) \phi_1(y) + y O(\tau^4) = c_0 y \alpha(x)+y
O(\tau^4)$, and so
\[
\hat{\Phi}_y = u_y \left[ 1- \frac{c_0}{k} \alpha(x) + O(\epsilon) + O(y) \right]
+ \alpha(X)(c_0 + O(y)).
\]
Since $s<0$, formula \eqref{eqn:laminar} implies that $u_y < 0$ near the bottom,
whereas the expression in the square brackets is positive. Moreover, it was
established in the proof of Theorem~3 that the function $\alpha(x)$ attains
negative values for some $x$. Taking this into account, the last formula shows
that $\hat{\Phi}_y$ attains negative values near the bottom. More careful but
simple analysis yields that the set of streamlines of $\hat{\Phi}$ has the
structure shown in Figure 1 and it is essentially the same for $\psi$.

\vspace{2mm}
	
\noindent {\bf Acknowledgements.} V.~K. was supported by the Swedish Research
Council (VR), 2017-03837. N.~K. acknowledges the support from the Link\"oping
University.

\bibliographystyle{jfm}
\bibliography{bibliography}	
	
\end{document}